\RequirePackage{amsmath}

\newif\ifsubmit     
\newif\ifllncs      
\newif\ifexabs      
\newif\ifblind

\submittrue

\ifllncs
  \documentclass[runningheads,a4paper]{llncs}

\else
  \documentclass[letterpaper,11pt,pdfa]{article}
  \usepackage[in]{fullpage}
\fi

\usepackage{iftex}
\ifPDFTeX
  \usepackage[utf8]{inputenc}
  \usepackage[noTeX]{mmap}
  \usepackage[T1]{fontenc}
\fi
\ifLuaTeX
  \usepackage{luatex85}
  \usepackage[noTeX]{mmap}
\fi

\usepackage{mdframed}
\usepackage{amsmath}
\usepackage{amsfonts}
\usepackage{amssymb}
\usepackage{amsthm}
\usepackage{color}
\usepackage{hhline}

\usepackage{appendix}
\usepackage{algorithm}
\usepackage{algpseudocode}
\usepackage{braket}
\usepackage{comment}
\usepackage{url}
\usepackage{bbm}
\usepackage{multicol}
\usepackage{mdframed}
\usepackage{multirow}

\usepackage{breakcites}

\usepackage[bookmarks]{hyperref}
\usepackage[nameinlink]{cleveref}
\hypersetup{
  colorlinks,
  allcolors=blue
}

\usepackage{tikz}

\ifllncs

  \spnewtheorem{claim}{Claim}{\bfseries}{\rmfamily}
  \crefname{claim}{claim}{claims}
  \Crefname{claim}{Claim}{Claims}
\else
  \newtheorem{theorem}{Theorem}[section]
  \newtheorem{definition}[theorem]{Definition}
  \newtheorem{remark}[theorem]{Remark}
  \newtheorem{lemma}[theorem]{Lemma}

  \newtheorem{claim}[theorem]{Claim}
  \newtheorem*{remark*}{Remark}

  \newtheorem*{theorem*}{Theorem}
  \newtheorem*{lemma*}{Lemma}

\usepackage[style=alphabetic,minalphanames=3,maxalphanames=4,maxnames=99,backref=true]{biblatex}

  \DeclareFieldFormat{eprint:iacr}{Cryptology ePrint Archive: \href{https://ia.cr/#1}{\texttt{#1}}}
  \DeclareFieldFormat{eprint:iacrarchive}{Cryptology ePrint Archive: \href{https://eprint.iacr.org/archive/#1}{\texttt{#1}}}
  \AtEveryBibitem{
    \clearlist{address}
    \clearfield{date}
    \clearfield{isbn}
    \clearfield{issn}
    \clearlist{location}
    \clearfield{month}
    \clearfield{series}
 
    \ifentrytype{book}{}{
      \clearlist{publisher}
      \clearname{editor}
    }
  }
\fi

\usepackage{appendix}
\usepackage{algorithm}
\usepackage{algpseudocode}
\usepackage{braket}
\usepackage{comment}
\usepackage{url}
\usepackage{multicol}
\usepackage{tikz}
\usetikzlibrary{arrows,automata,positioning}
\usepackage{caption}
\usepackage[caption=false]{subfig}
\usepackage[font=small,labelfont=bf]{caption}
\usepackage{comment}
\usepackage{pifont}
\usetikzlibrary{patterns}

\usepackage{enumitem}
\setlist[description]{noitemsep}
\setlist[enumerate]{noitemsep}
\setlist[itemize]{noitemsep}

\usepackage{soul, xcolor, xparse}
\makeatletter
  \ExplSyntaxOn
    \cs_new:Npn \white_text:n #1
    {
      \fp_set:Nn \l_tmpa_fp {#1 * .01}
      \llap{\textcolor{white}{\the\SOUL@syllable}\hspace{\fp_to_decimal:N \l_tmpa_fp em}}
      \llap{\textcolor{white}{\the\SOUL@syllable}\hspace{-\fp_to_decimal:N \l_tmpa_fp em}}
    }
    \NewDocumentCommand{\whiten}{ m }
    {
      \int_step_function:nnnN {1}{1}{#1} \white_text:n
    }
  \ExplSyntaxOff
  
  \NewDocumentCommand{ \varul }{ D<>{5} O{0.2ex} O{0.1ex} +m } {%
    \begingroup
    \setul{#2}{#3}%
    \def\SOUL@uleverysyllable{%
      \setbox0=\hbox{\the\SOUL@syllable}%
      \ifdim\dp0>\z@
      \SOUL@ulunderline{\phantom{\the\SOUL@syllable}}%
      \whiten{#1}%
      \llap{%
        \the\SOUL@syllable
        \SOUL@setkern\SOUL@charkern
      }%
      \else
      \SOUL@ulunderline{%
        \the\SOUL@syllable
        \SOUL@setkern\SOUL@charkern
      }%
      \fi}%
    \ul{#4}%
    \endgroup
  }
\makeatother

\usepackage{mleftright}

\newcommand{\norm}[1]{\left\lVert#1\right\rVert}

\newcommand{\As}{\mathcal{A}}
\newcommand{\Bs}{\mathcal{B}}

\newcommand{\Sim}{{\sf Sim}}

\newcommand{\cA}{\mathcal{A}}

\newcommand{\cB}{\mathcal{B}}

\newcommand{\cC}{\mathcal{C}}

\newcommand{\cR}{{\mathcal{R}}}

\mdfdefinestyle{figstyle}{ 
  linecolor=black!7, 
  backgroundcolor=black!7, 
  innertopmargin=10pt, 
  innerleftmargin=25pt, 
  innerrightmargin=25pt, 
  innerbottommargin=10pt 
}

\renewcommand{\kappa}{\ell}

\newcommand{\ketbra}[2]{\ket{#1}\!\bra{#2}}

\newcommand\numeq[1]%
  {\stackrel{\scriptscriptstyle(\mkern-1.5mu#1\mkern-1.5mu)}{=}}

\newcommand{\xo}[1]{{x}_{{\sf o}, #1}}
\newcommand{\yo}[1]{{y}_{{\sf o}, #1}}
\newcommand{\xout}{\vec{x}_{\sf o}}
\newcommand{\yout}{\vec{y}_{\sf o}}

\newcommand{\cX}{\mathcal{X}}

\newcommand{\cY}{\mathcal{Y}}

\newcommand{\cZ}{\mathcal{Z}}

\newcommand{\multigame}{\allowbreak multi-output $k$-search game } %

\newcommand{\gamemath}{\mathcal{G}}

\newlength{\saveparindent}
\newlength{\saveparskip}
\newcounter{ctr}

\newcounter{ectr}

\newcommand{\ignore}[1]{}

\newcommand{\sym}{\mathcal{S}} 

\title{
Improved Quantum Lifting by Coherent Measure-and-Reprogram
}

\author{Alexandru Cojocaru \thanks{University of Edinburgh}
\and Juan Garay \thanks{Texas A\&M University}
\and Qipeng Liu \thanks{UC San Diego}
\and Fang Song \thanks{Portland State University}
}

\date{}

\begin{document}

\maketitle

\begin{abstract}
We give a tighter lifting theorem for security games in the quantum random oracle model. At the core of our main result lies a novel measure-and-reprogram framework that we call {\em coherent reprogramming}. This framework gives a tighter lifting theorem for query complexity problems, that only requires purely classical reasoning. As direct applications of our lifting theorem, we first provide a quantum direct product theorem in the average case --- i.e., 
an enabling tool to determine the hardness of solving multi-instance security games.  
This allows us to derive in a straightforward manner the hardness of  
various security games, for example (i) the non-uniform hardness of salted games, 
(ii) the hardness of specific cryptographic tasks such as the multiple instance version of one-wayness and collision-resistance, and (iii) uniform or non-uniform hardness of many other games. 
\end{abstract}

 \newpage
 \tableofcontents

\section{Introduction}

Hash functions are a fundamental workhorse in modern cryptography. Efficient constructions such as SHA-2 and SHA-3 are widely used in real-world cryptographic applications. To facilitate the analysis of constructions based on hash functions, Bellare and Rogaway~\cite{bellare1993random} proposed a framework known as the random oracle model (ROM). Recent development of quantum computing demands re-examining security against potential quantum attackers. The \emph{quantum random oracle model} (QROM) has since been proposed by Boneh {\em et al.}~\cite{AC:BDFLSZ11} as an extension of ROM by taking into account quantum attackers. Various techniques have been developed for analyzing security in the QROM; however, they are often either {\em ad-hoc} (for specific scenarios) or too involved to apply. 

In this paper we revisit a general tool for lifting security 
from ROM to QROM by
Yamakawa and Zhandry~\cite{YZ21}. The lifting theorems are applicable to search games in (Q)ROM between a classical challenger interacting with an adversary (e.g., think of an adversary that aims to find a preimage of 0 in the random oracle, 
with the challenger 
querying the random oracle to verify the adversary's answer). Roughly speaking, the lifting theorems assert that if a search game with a challenger performing a constant number of queries to the random oracle is hard against a classical adversary, then it is also hard against a quantum adversary in the QROM. Specifically, if the challenger performs $k$ queries and if a quantum adversary performs $q$ quantum queries and wins the search game with probability $\epsilon$, then there exists a classical adversary performing only $k$ classical queries winning with probability $\epsilon / (2q +1)^{2k}$. 

This tool is particular powerful to establish  quantum query lower bounds in the QROM. Let us consider function inversion from above for example; the goal is to find an input $x$, whose image equals $0$. In this case, $k = 1$ and 
$$\frac{\epsilon}{(2q+1)^{2k}} = \frac{\epsilon}{(2q+1)^{2}} \leq \frac{1}{N}.$$ 
This is because a single query reveals a pre-image of $0$ with probability at most $1/N$. Therefore we have
$\epsilon \leq {(2q+1)^2} /{N}$, 
which immediately reproduces the tight bound for the famous Grover's search problem~\cite{BBBV97}. However, for a $k$-search problem whose goal is to find $k$ distinct inputs that all map to $0$, the bound derived from~\cite{YZ21} is $O\left({(k q)^2}/{N} \right)^k$ for any quantum algorithm with $k q$ queries, which has a large $k^{2k}$ factor gap from the tight bound $\Theta\left({q^2}/{N} \right)^k$
\footnote{
 {We believe this is a folklore result that to our knowledge, this bound follows from a result in \cite{CGKSW23} (Theorem 3.1).
    Moreover, we would like to emphasize that since our main result is a strengthening of the lifting Lemma of \cite{YZ21}, we can also show that our result concerning the bound of this problem is stronger than the bound derived from \cite{YZ21}.}
}
. Similar weaknesses appear in a variety of problems involving multiple inputs. 

In this work, we derive a new tighter lifting theorem for search games. If the challenger performs $k$ queries and the quantum adversary performs $q$ quantum queries and wins the search game with probability $\epsilon$, then there exists a quantum adversary performing only $k$ quantum queries and winning with probability $\epsilon /{2^{2k}{{q + k} \choose k}^2}$, improving 
on the previous lifting theorem by Yamakawa and Zhandry. Let us consider a $(k q)$-quantum-query algorithm for the previous $k$-search problem.  Our bound (in this case, $q$ in the theorem should be $k q$) gives the tight bound as below:
\begin{align*}
    \frac{\epsilon}{{2^{2k}{{k q + k} \choose k}^2}} \leq \frac{1}{N^k} \quad\longleftrightarrow\quad \epsilon \leq \frac{2^{2k} \binom{k q + k}{k}}{N^k}^2 \leq O\left( \frac{(q+1)^2}{N} \right)^k \, .
\end{align*}

To achieve this, we develop a new measure-and-reprogram technique 
which is a key technical contribution of our work.
The technique, which we call  
\emph{coherent reprogramming}, 
improves on the recent results on adaptively reprogramming QRO on multiple points
by Don {\em et al.}~\cite{DFM20,DFMS22} and Liu and Zhandry \cite{LZ19},
yielding
tighter reprogramming bounds than the existing measure-and-reprogram proofs. 
As an immediate consequence, we are able to derive tighter quantum hardness bounds 
for many applications, such as  direct-product theorems, salted security, and non-uniform security.  

\subsection{Summary of Our Results}

\noindent{\bf Lifting Theorem for Search Games.}
Our central result is a new lifting theorem for search games that relates (and upper bounds) the success probability of an arbitrary quantum algorithm to the success probability of a quantum algorithm performing a small number of queries to the RO.
More formally:

\begin{theorem}[Quantum Lifting Theorem (Informal)]\label{thm:intro_main1}
Let $\gamemath$ be a search game with a classical challenger $\cC$ that 
performs at most $k$ queries to the RO,
and let $\cA$ be a $q$-quantum query adversary in the game $\gamemath$ (against the $k$-classical query challenger $\cC$).
Then there exists a $k$-quantum query adversary $\cB$ 
such that:
\begin{equation*}
\Pr[\cB \textrm{ wins } \gamemath] \geq \frac{1}{{2^{2k}{{q + k} \choose k}^2}} \Pr[\cA \text{ wins } \gamemath].
\end{equation*}
\end{theorem}
\begin{remark}
    Comparing to the lifting theorem in \cite{YZ21}, we have a better loss $2^{2k} \binom{q+k}{k}^2$, whereas it is $(2q+1)^k$ in their work. Since the algorithm often makes more queries than the challenger $q \gg k$, it is roughly a $(k!)^2$ save. In \cite{YZ21}, they are able to reduce a $q$-quantum-query algorithm to a $k$-classical-query algorithm; whereas in this work, we only reduce the number of the queries, with the algorithm $\Bs$ still making quantum queries. Nonetheless, it does not affect the applications and improvement we have in this work. {Our framework handles the case where the challenge is independent of the oracle (similar to the results in \cite{YZ21}). We leave the case of oracle-dependent challenges as an interesting open question.}
\end{remark}

\medskip
\noindent{\bf Coherent Reprogramming.}
At the core of our main lifting result above lies a new framework for quantum reprogramming which we call \emph{coherent reprogramming}. This new 
framework has the following 
advantages:
\begin{enumerate}
    \item It simplifies the proofs and frameworks of existing quantum reprogramming results;
    \item It yields improved tighter reprogramming bounds; and
    \item It implies in a straightforward manner several applications in quantum query complexity and cryptography.
\end{enumerate}

In order to present our main coherent reprogramming result, we first need to introduce 
a few notions. For an oracle $H$ we  call $H_{x, y}$ the reprogrammed oracle that behaves almost like the original function $H$, with the only difference that its value on input $x$ will be $y$. Similarly, we 
define the reprogrammed oracle on $k$ inputs $\vec{x} = (x_1, ..., x_k)$ and $k$ corresponding outputs $\vec{y} = (y_1, ..., y_k)$, denoted by $H_{\vec{x}, \vec{y}}$, as the original function $H$ with the only difference that for every input $x_i$ in $\vec{x}$, the corresponding image will be $y_i$ in $\vec{y}$.

\begin{theorem} [Coherent Reprogramming (Informal)]
    \label{thm:reprogramming_informal}
    Let $H, G$ be two random oracles. 
    Let $\As$ be any $q$-quantum query algorithm to the oracle $H$, and let $\xout = (x_1, ..., x_k) \in X^k$ be any $k$-vector of inputs and $\yout = (y_1, ..., y_k) = (G(x_1), ..., G(x_k))$.
    Then there exists a simulator algorithm $\Sim$ that given oracle access to $H, G$ and $\cA$, simulates the output of $\cA$ having oracle access to $H_{\xout, \yout}$ (the  reprogrammed version of $H$) with probability: 
    \begin{align*}
     \Pr_{H, G}\left[\Sim \text{ outputs correct } (\vec{x}, \vec{y})  \right]
       \geq  \frac{1}{{2^{2k}{{q + k} \choose k}^2}}  \cdot \Pr_{H, G}\left[\cA^{H_{\xout, \yout}} \text{ outputs correct } (\vec{x}, \vec{y}) \right]. 
    \end{align*}
    where "correct" 
    is defined with respect to some predicate that can depend on the reprogrammed oracle $H_{\xout, \yout}$.
\end{theorem}
\begin{remark}
    Similar to the comparison between \Cref{thm:intro_main1} and \cite{YZ21}, the second theorem improves the factor $(2q+1)^k$ in \cite{DFM20} to $2^{2k} \binom{q+k}{k}^{2}$. Our simulator does not measure and reprogram directly, but does everything coherently (or in superposition). 
\end{remark}

Next, we show the applications of our lifting theorem 
in query complexity and cryptography.

\medskip
\noindent\textbf{Quantum Lifting Theorem with Classical Reasoning.} 
A \multigame between a challenger and an adversary is defined as follows. The adversary receives $k$ different challenges from the challenger, and at the end of their interaction, the adversary needs to respond with $k$ outputs. If the $k$ outputs (taken together) satisfy some relation $R$ specified by the game, we say the adversary wins the \multigame.
The goal of the \emph{lifting theorem} is to establish the hardness of solving the \multigame by any general quantum adversary, with only simple classical reasoning. For an arbitrary $k$-ary relation $R$, let $\sym_k$ be the symmetric group on $[k]$ and we define: 
\begin{align*}
p(R): = \Pr[\exists \pi\in \sym_k \ | \ (y_{\pi(1)}, y_{\pi(2)}, ..., y_{\pi(k)}) \in R  : (y_1, ..., y_k) \xleftarrow{\$} Y^k] \,.
\end{align*}
Note that $p(R)$ is a quantity that only depends on the game itself, and can be calculated with only classical reasoning. 

\begin{theorem}[Quantum Lifting Theorem with Classical Reasoning (Informal)]
    For any quantum algorithm $\cA$ equipped with $q$ quantum queries to a random oracle $H : X \rightarrow Y$, $\cA$'s success probability to solve the \multigame as specified by the winning relation $R$, is bounded by:
    \begin{align*}
        \Pr[\cA \text{ wins \multigame}] \leq {2^{2k}{{q + k} \choose k}^2} \cdot p(R) \, .
    \end{align*}
\end{theorem}

Our lifting theorem translates into
the following quantum hardness results for our applications in query complexity and cryptography.

\medskip
\noindent\textbf{Direct Product Theorem.} 
We give the first direct product theorem (DPT) in the average case (in the QROM). Previously, only worst-case quantum DPTs were known \cite{sherstov2011strong,lee2013strong} and were proof-method dependent; until recently, Dong et al.~\cite{dong2024salting} shows the first average-case quantum DPTs for some problems in the QROM. While they are non-tight, our DPTs works for all games in the QROM and proof-method independent.

Our direct product theorem establishes the hardness of solving $g$ independent instances (each instance is associated with an independent oracle) of a game $\gamemath$ given a total of $g \cdot q$ quantum queries: 
\begin{theorem}[Direct Product Theorem]
    For any quantum algorithm $\As$ equipped with $g \cdot q$ quantum queries, $\As$'s success probability to solve the Direct Product game $\gamemath^{\otimes g}$ with the underlying $\gamemath$ specified by the winning relation $R$, is bounded by
\begin{align*}
    \Pr[\As & \text{ wins } \gamemath^{\otimes g}] \leq 
     \left( {2^{2k}{{q + k} \choose k}^2}  p(R) \right)^g.
\end{align*}
\end{theorem}

\medskip
\noindent\textbf{Non-uniform Security of Salting.} 
The above theorem directly implies non-uniform security of salting. 
Non-uniform attacks allow a malicious party to perform heavy computation offline and attack a protocol much more efficiently, using the information in the offline stage. 
Salting is a generic method that prevent non-uniform attacks against hash functions. Chung et al.~\cite{chung2020tight} shows that ``salting generically defeats quantum preprocessing attacks''; they show that if a game in the QROM is $\epsilon(q)$ secure, the salted game with salt space $[K]$ is $\epsilon(q) + \frac{Sq}{K}$ secure against a quantum adversary with $S$-bit of advice. Their bound is non-tight, since when the underlying game is collision-finding, the tight non-uniform security should be $\epsilon(q) + \frac{S}{K}$. Improving the additive factor is an interesting open question and until recently \cite{dong2024salting} is able to answer this question affirmatively for a limited collection of games.

Using our direct product theorem, we show:
\begin{theorem}[Another ``Salting Defeats Quantum Preprocessing'']
For any non-uniform quantum algorithm $\As$ equipped with $q$ quantum queries and $S$-bit of classical advice, $\As$'s success probability to solve the salted game $\gamemath_s$ with the underlying $\gamemath$ specified by the winning relation $R$, is bounded by
\begin{align*}
    \Pr[\As & \text{ wins } \gamemath_s] \leq 
      4 \cdot \left( {2^{2k}{{q + k} \choose k}^2}  p(R) + \frac{S}{K} \right).
\end{align*}
\end{theorem}
Our bound is incomparable to that in \cite{chung2020tight}. We are able to improve the additive term from $\frac{Sq}{K}$ to $\frac{S}{K}$, while only able to give an upper bound for $\epsilon(q)$. Even our bound is non-tight in general, it still confirms (on a high level) that the help from classical advice only comes from the following:
\begin{itemize}
    \item using $S$-bit advice to store solutions for $S$ random salts;
    \item if the challenge salt matches with the random salts (with probability $\frac{S}{K}$), the attack succeeds; otherwise, proceed the attack as if there is no advice. 
\end{itemize}

\medskip
\noindent{\bf Non-Uniform Security.}
By combining our lifting theorem with the results 
by Chung {\em et al.}~\cite{chung2020tight}, we derive the following results concerning the security (hardness) against non-uniform quantum adversaries with classical advice, for a broader class of games. 

\begin{lemma}[Security against Quantum Non-Uniform Adversaries (Informal)]
Let $\gamemath$ be any classically verifiable search game specified by the winning relation $R$. Let $R^{\otimes S}$ be the winning relation of the multi-instance game of $\gamemath$. 
Any quantum non-uniform algorithm $\cA$ equipped with $q$ quantum queries and $S$ classical bits of advice, can win 
$\gamemath$ with probability at most:
\begin{align*}
   \Pr[\cA \text{ wins } \gamemath] \leq 4 \cdot 2^{2k} {{S(q+k)} \choose S k}^{\frac{2}{S}} \cdot  p(R^{\otimes S})^{1/S} \, .
\end{align*}
\end{lemma}

\medskip

To demonstrate the power of our results, we also apply them to the hardness of three concrete cryptographic tasks: the multiple instance versions of one-wayness, collision resistance and search, as described next. Note that the applications we list below are non-exhaustive, given $p(R)$ is easy to define for almost every game.

\medskip
\noindent{\bf Hardness of Multi-Image Inversion.} Firstly, we can analyze the quantum hardness of inverting $k$ different images of a random oracle $H : [M] \rightarrow [N]$.

Our first result establishes the quantum hardness of multi-image inversion, which is a tight bound as already proven in \cite{chung2020tight}, but achieved here in a much simpler way, directly from our quantum lifting theorem.

\begin{lemma}[Quantum Hardness of Multi-Image Inversion (Informal)]
    For any distinct $\vec{y} = (y_1, ..., y_k) \in [N]^k$ and for any $q$-quantum query algorithm $\cA$ whose aim is to invert all the images in $\vec{y}$, the success probability of $\cA$ is upper bounded by:
     \begin{equation*}
         \Pr_H[\cA(\vec{y}) \rightarrow \vec{x} = (x_1, ..., x_k) \ : \ H(x_i) = y_i \ \forall i \in [k]] \leq {2^{2k}{{q + k} \choose k}^2} \cdot \frac{k!}{N^k}.
     \end{equation*}
     
\end{lemma}

\medskip
\noindent{\bf 
Hardness of Multi-Collision Finding.}  Secondly, we can analyse the quantuanalyzeess of finding $k$ collisions, namely, $k$ inputs that map to the same image of the random oracle $H : [M] \rightarrow [N]$. We can also determine upper bounds for solving the salted version of this task, as well as the hardness of finding a collision for quantum algorithms that are also equipped with advice.

\begin{lemma}[Quantum Hardness of Multi-Collision Finding and Salted Multi-Collision Finding (Informal)]
    For any $q$-quantum query algorithm $\cA$, 
    the probability of solving the $k$-multi-collision problem is at most:
    \begin{equation*}
        \Pr_H[\cA() \rightarrow \vec{x} = (x_1, ..., x_k) \ : \ H(x_1) = ... = H(x_k)] \leq 
         \frac{1}{N^{k - 1}} \left[\frac{2e(q+k)}{k}\right]^{2k}.
    \end{equation*}
    
 Any quantum algorithm $\cA$ equipped with $q$ quantum queries and $S$-bit of classical advice can win the salted multi-collision finding game with salted space $[K]$ with probability at most:
 \begin{equation*}
     \Pr_H[\cA() \rightarrow \vec{x} = (x_1, ..., x_k) \ : \ H(x_1) = ... = H(x_k)] \leq \frac{4}{N^{k - 1}} \left[\frac{2e(q+k)}{k}\right]^{2k} + \frac{4S}{K}.
 \end{equation*}
\end{lemma}
The above bounds become $O(q^4/N)$ and $O(q^4/N +S/K)$ respectively, for $k = 2$ (the standard collision-finding). Previous work \cite{YZ21} achieves the same uniform bound, but only achieves $O((Sq)^4/N +S/K)$, due to the extra loss in their lifting theorem.

\medskip
\noindent{\bf Hardness of Multi-Search.}
Finally, we also establish a tight bound  
for finding $k$ distinct inputs that all map to $0$ under the random oracle $H : [M] \rightarrow [N]$. This is potentially useful in analyzing proofs-of-work in the blockchain context~\cite{garay2015bitcoin}.

\begin{lemma}[Quantum Hardness of Multi-Search]\
    For any $q$-quantum query algorithm $\cA$ whose task is to find $k$ different preimages of $0$ of the random oracle, the success probability of $\cA$ is upper bounded by:
    \begin{align*}
          \Pr_H[\cA() \rightarrow \vec{x} = (x_1, ..., x_k) \ : \ H(x_i) = 0 \ \forall i \in [k]] 
             \leq 
             \left[\frac{4e^2 (q + k)^2}{Nk^2} \right]^k.
    \end{align*}
\end{lemma}

\subsection{Related Work}

The measure-and-reprogram framework 
was proposed, and subsequently generalized and improved with tighter bounds 
in the works of \cite{DFM20,LZ19,DFMS22,GHHM21}. A main application of the framework has been in the context of the Fiat-Shamir transformation, 
with several works 
establishing its post-quantum security 
\cite{Chailloux19,DFMS19,AFK22,AFKR23,GOPTT23}. Other cryptographic applications of measure-and-reprogram  have been considered in \cite{Katsumata21,BKS21,ABKK23,JMZ23,JSYP23,KX24}. Finally, applications in query complexity of the framework have been developed in \cite{chung2020tight,YZ21}.

\section{Preliminaries}

\paragraph*{Notation.} For two vectors $\vec{x}, \vec{x}' \in X^k$, we say $\vec{x} \equiv \vec{x}'$ if and only if there exists a permutation $\sigma$ over the indices $\{1, 2, \ldots, k\}$ such that $x'_i = x_{\sigma(i)}$. For a function $H: X \to Y$ and $\vec{x} \in X^k$, $H(\vec{x})$ is defined as $(H(x_1), H(x_2), \ldots, H(x_k))$. We say $x \in \vec{x}$, if $x = x_i$ for some $i \in [k]$. 

\subsection{Quantum Query Algorithms}
\label{sec:prelim_quantum_query_algo}

We will denote a quantum query algorithm by $\cA$. Let $q$ be the total number of quantum queries of $\cA$. By $\cA^H$ we mean that $\cA$ has quantum access to the function $H$.

A quantum oracle query to $H$ will be applied as the unitary $O_H$: $O_H \ket{x} \ket{y} \longrightarrow \ket{x} \ket{y \oplus H(x)}$. 
Without loss of generality, we assume an algorithm will never perform any measurement (until the very end) and thus the internal state is always pure. 
We use $\ket {\phi_i^H}$ to denote the algorithm $\As$'s internal (pure) state right after the $i$-th query. 
\begin{align*}
    \ket {\phi^H_i} =  O_H U_i \cdots O_H U_1 \ket 0. 
\end{align*}
Specifically, we have, 
\begin{itemize}
    \item $\ket {\phi_0^H} = \ket 0$ is the initial state of $\As$; 
    \item $\ket {\phi_q^H}$ is the final state of $\As$. 
\end{itemize}
Without loss of generality, the algorithm will have three registers $\cX, \cY, \cZ$ at the end of the computation, where $\cX$ consists of a list of inputs, $\cY$ consists of a list of outputs corresponding to these inputs and some auxiliary information in $\cZ$. 

\begin{definition}[Reprogrammed Oracle] Reprogram oracle $H$ to output $y$ on input $x$, results in the new oracle, defined as:
$$
H_{x, y}(z) = 
\begin{cases}
     y,  & \text{ if } z = x \\
     H(z), & \text{ otherwise.}
\end{cases}
$$
We can similarly define a multi-input reprogram oracle $H_{\vec{x},\vec{\Theta}}$ for $\vec{x} \in X^k$ without duplicate entries and $\vec{\Theta} \in Y^k$:
$$
H_{\vec{x}, \vec{\Theta}}(z) = 
\begin{cases}
     \Theta_i,  & \text{ if } z = x_i \\
     H(z), & \text{ otherwise.}
\end{cases}
$$
\end{definition}

\subsection{Quantum Measure-and-Reprogram Experiment}

We recall the measure-and-reprogram experiment and the state-of-the-art results here, first proposed by \cite{DFM20} and later adapted by \cite{YZ21}.

\begin{definition}[Measure-and-Reprogram Experiment]\label{def:classical_measure_and_reprogram}
Let $\cA$ be a $q$-quantum query algorithm that outputs $\vec{x} \in X^k$ and $z \in Z$. For a function $H : X \rightarrow Y$ and $\vec{y} = (y_1, ..., y_k) \in Y^k$, define a measure-and-reprogram algorithm $\cB[H, \vec{y}]$:
\begin{enumerate}
    \item For each $j \in [k]$, uniformly pick $(i_j, b_j) \in ([q] \times \{0, 1\}) \cup \{(\perp, \perp)\}$ such that there does not exist $j \neq j'$ such that $i_j = i_{j'} \neq \perp$;
    \item Run $\cA^O$ where the oracle $O$ is initialized to be a quantumly accessible classical oracle that computes $H$ and when $\cA$ makes its $i$-th query, the oracle is simulated as follows:
    \begin{enumerate}
        \item If $i = i_j$ for some $j \in [k]$, measure $\cA$'s query register to obtain $x_j'$ and do either of the following:
        \begin{enumerate}
            \item If $b_j = 0$, reprogram $O$ using $(x_j', y_j)$ and answer $\cA$'s $i_j$-th query using the reprogrammed oracle;
            \item If $b_j = 1$, answer $\cA$'s $i_j$-th query using oracle before reprogramming and then reprogram $O$ using $(x_j', y_j)$;
        \end{enumerate}
        \item Else, answer $\cA$'s $i$-th query by just using the oracle $O$ without any measurement or reprogramming;
    \end{enumerate}
    \item Let $(\vec{x} = (x_1, ..., x_k), z)$ be $\cA$'s output;
    \item For all $j \in [k]$ such that $i_j = \perp$, set $x_j' = x_j$
    \item Output $\vec{x'} := ((x_1', ..., x_k'), z)$
\end{enumerate}
\end{definition}

We next state the current state-of-the-art quantum measure-and-reprogram result.

\begin{lemma}[Quantum Measure-and-Reprogram (adaptation from \cite{DFM20,YZ21})]
For any $H : X \rightarrow Y$, for any $\vec{x}^* = (x_1^*, ..., x_k^*) \in X^k$ without duplicated entries, for all $\vec{y} = (y_1, ..., y_k)$ and any relation $R \subseteq X^k \times Y^k \times Z$, we have:
\begin{equation*}
    \begin{split}
        \Pr&[\vec{x}' = \vec{x}^* \wedge (\vec{x}', \vec{y}, z) \in R \ | \ (\vec{x}', z) \leftarrow \cB[H, y]] \\ 
        & \geq \frac{1}{(2q+1)^{2k}} \Pr[\vec{x} = \vec{x}^* \wedge (\vec{x}, \vec{y}, z) \in R \ | \ (\vec{x}, z) \leftarrow \cA^{H_{\vec{x}^*, \vec{y}}}]
    \end{split}
\end{equation*}
    where $\cB[H, y]$ is the measure-and-reprogram experiment.  
\end{lemma}

\subsection{Predicates and Success Probabilities}
\begin{definition}[Predicate/Verification Projection/Symmetric Predicate] 
Let $R$ be a relation on $X^k \times Y^k \times Z$. A predicate $V^H(\vec{x}, \vec{y}, z)$ parameterized by an oracle $H$, returns $1$ if and only if $(\vec{x}, \vec{y}, z) \in R$ and $H(x_i) = y_i$ for every $i \in \{1, 2, \ldots, k\}$. 

Let $\cX, \cY, \cZ$ be the registers that store $\vec{x}, \vec{y}, z$, respectively. 
We define $\Pi^H_V$ as the projection corresponding to $V^H$: 
\begin{equation*}
    \Pi^H_V \ket {\vec{x}, \vec{y}, z} = 
    \begin{cases}
        \ket {\vec{x}, \vec{y}, z} & \text{ if } V^H(\vec{x}, \vec{y}, z) = 1\\
        0 & \text{otherwise}
    \end{cases}. 
\end{equation*}

\end{definition}

\medskip

Finally, for any predicate $V^H$, we are able to establish the success probability using the projection $\Pi^H_V$. 
\begin{definition}[Success Probability] \label{def:succ_prob}
Let $\cA$ a quantum query algorithm. Its success probability of outputting $\vec{x}, \vec{y}, z$ such that $H(\vec{x}, \vec{y}, z) = 1$ is defined by 
\begin{equation*}
    \Pr\left[\cA^H \rightarrow (\vec{x}, \vec{y}, z) \text{ and } V^H(\vec{x}, \vec{y}, z) = 1\right] = \norm{\Pi^H_V \ket{\phi_q^H}}^2. 
\end{equation*}
(Recall that $\ket{\phi_q^H}$ is the final state of $\As$.) 

Sometimes, we care about the event that $\As$ outputs a particular $\vec{x}$ and still succeeds. 
For any $\xout$, the following probability denotes that $\As$ outputs $\vec{x} \equiv \xout$  and succeeds: 
\begin{equation*}
    \Pr\left[\cA^H \rightarrow (\vec{x}, \vec{y}, z) ,\ \, \ \vec{x} \equiv \xout \,\text{ and } V^H(\vec{x}, \vec{y}, {z}) = 1 \right] = \norm{G_{\xout} \Pi^H_V \ket{\phi_q^H}}^2,
\end{equation*}
where $G_{\xout}$ is defined as the projection that checks whether $\cA$ consists of $\vec{x} \equiv \xout$. 
\end{definition}

\section{Coherent Measure-and-Reprogram}

In this section, we give our main theorem: the coherent measure-and-reprogram theorem. 
A main difference between our theorem and the previous measure-and-reprogram theorem~\cite{DFM20} 
is that our simulator needs to make quantum queries, instead of classical queries, which is potentially required by the coherent nature of our simulator and gives tighter reprogramming bounds for many applications. While this
makes the simulator slightly more complicated, 
it yields improved bounds on the various applications that we 
mention 
in the next section. 

\subsection{Main Theorem}

We give our main theorem below. 
\begin{theorem}\label{thm:coherent_measure_and_reprogram}
    Let $H, G: \{0,1\}^m \to \{0,1\}^n$ be two functions $X \to Y$. 
    Let $k$ be a positive integer (can be a computable function in both $n$ and $m$). There exists a black-box quantum algorithm $\Sim^{H, G, \As}$, satisfying the  properties below. 
    Let $V^H$ be any predicate defined over $X^k \times Y^k \times Z$. 
    Let $\As$ be any $q$-quantum query algorithm to the oracle $H$. Then for any $\xout \in X^k$ without duplicate entries and $\yout = G(\xout)$, we have, 
    \begin{align*}
        & \Pr_{H, G}\left[\Sim^{H, G, \As} \rightarrow (\vec{x}, \vec{y}, z)  \text{ and } \vec{x} \equiv \xout \text{ and } V^{H_{\xout, \yout}}(\vec{x}, \vec{y}, z) = 1\right]  \\ 
        \geq & \frac{1}{2^{2k} \binom{q+k}{k}^{2}}  \cdot \Pr_{H, G}\left[\cA^{H_{\xout, \yout}} \rightarrow (\vec{x}, z) \text{ and } \vec{x} \equiv \xout  \text{ and } V^{H_{\xout, \yout}}(\vec{x}, H_{\xout, \yout}(\vec{y}), z) = 1\right]. 
    \end{align*}
    Furthermore, $\Sim$ makes exactly $k$ quantum queries to $G$ and has a running time polynomial in $n, m, k$ and the running time of $\As$. 
\end{theorem}

\medskip

Before formally defining our simulator, we introduce one more notation: controlled reprogrammed oracle queries. That is, an oracle query will be reprogrammed by a list of input and output pairs in a control register. 
\begin{definition}[Controlled Reprogrammed Oracle Query]
    For every $x \in X, y \in Y$, every $\ell > 0$ and $\vec{x} \in X^\ell$ without duplicated  entries, $\Theta \in Y^\ell$, controlled reprogrammed oracle $O^{\sf ctrl}_H$ acts as below. 
    \begin{align*}
        O^{\sf ctrl}_H \ket {x} \ket y \ket {\vec{x}, \vec{\Theta}} = \left( O_{H_{\vec{x},\vec{\Theta}}} \ket x \ket y \right) \ket {\vec{x}, \vec{\Theta}}
    \end{align*}
    Its behavior on $\ket {\vec{x}}$ with duplicated entries can be arbitrarily defined as long as unitarity is maintained, as this case will never occur in the simulator or our analysis. 
\end{definition}

We define our simulator used in \Cref{thm:coherent_measure_and_reprogram}, as follows:
\begin{definition}[Coherent Measure-and-Reprogram Experiment]
\label{def:quantum_simulator}
Let $\cA$ be a $(q+k)$-quantum query algorithm that outputs $\vec{x} = (x_1, \ldots, x_k) \in X^k, \vec{y} \in Y^k$ and $z \in Z$.
We assume $\vec{y}$ is always computed by $H(\vec{x})$, using the last $k$ queries. 
For a function $H : X \rightarrow Y$ and $G:  X \rightarrow Y$, define a measure-and-reprogram algorithm $B[H, G]$:
\begin{enumerate}
    \item Pick a uniformly random subset $\vec{v}$ of $[q+k]$, of length $k$. We have $1 < v_1 < \cdots < v_k \leq q + k$. 
    Pick $\vec{b} \in \{0,1\}^k$ uniformly at random. 
    \item Run $\cA$ \ul{with an additional control register ${\cal R}$, initialized as empty $\ket \emptyset$}. Define the following operation $U$ that updates the control register: for $x$ that is not in $L$, 
    \begin{align*}
        U \ket x {\ket L}_{\cal R} \gets \ket x {\ket {L \cup (x, G(x))}}_{\cal R}. 
    \end{align*}
    Here $L$ is the set of input and output pairs. Since we will only work with basis states $\ket x \ket L$ whose $x$ is not in $L$, $U$ clearly can be implemented by a unitary (by assuming that the list is initialized as empty). 
    
    \item When $\cA$ makes its $i$-th query, 
    \begin{enumerate}
        \item If $i = v_j$ for some $j \in [k]$, do either of the following:
        \begin{enumerate}
            \item If $b_j = 0$, \ul{update ${\cal R}$ using the input register and $G$}, and make the $i$-th query to $H$ controlled by ${\cal R}$ (see $O_H^{\sf ctrl}$ above); 
            \item If $b_j = 1$, make the $i$-th query to $H$ controlled by ${\cal R}$ and \ul{update ${\cal R}$ using the input register and $G$}. 
            \item Before updating the control register, it \ul{checks coherently that the input register is not contained in the control register}; otherwise, it aborts.
        \end{enumerate}
        \item Else, answer $\cA$'s $i$-th query controlled by ${\cal R}$;
    \end{enumerate}
    \item Let $(\vec{x}, \vec{y}, z)$ be $\cA$'s output;
    \item \ul{Measure ${\cal R}$ register to obtain $L = (\vec{x}', \vec{\Theta}')$}. 
    \item Output $(\vec{x}, \vec{y}, z)$ \ul{if $\vec{x}' \equiv \vec{x}$}; otherwise, abort. 
\end{enumerate}
\end{definition}
At a high level, our simulator resembles that in \Cref{def:classical_measure_and_reprogram}; instead of measuring $\As$'s queries, we put it into a separate register (a.k.a., measure the queries coherently). With the ``controlled reprogrammed oracle query'', we are still able to progressively reprogram the oracle and run the algorithm with (an) updated oracle(s). The ability of coherently measuring and reprogramming, makes all the improvement (mentioned in the later sections) possible. 

\begin{proof}[Proof of \Cref{thm:coherent_measure_and_reprogram}]
Before we start with the proof, we first recall and introduce some notations. Fix any $\vec{x} \in X^k$ without duplicate entries and $\vec{\Theta} \in Y^k$.  
Recall that in \Cref{sec:prelim_quantum_query_algo}, $\left| {\phi^{H_{\vec{x}, \vec{\Theta}}}_q} \right\rangle$ is the state of the algorithm $\As$ after making all its queries to $H_{\vec{x}, \vec{\Theta}}$. More precisely, it is: 
\begin{align*}
    \left| {\phi^{H_{\vec{x}, \vec{\Theta}}}_q} \right\rangle=  O_{H_{\vec{x}, \vec{\Theta}}} U_q \cdots O_{H_{\vec{x}, \vec{\Theta}}} U_1 \ket 0. 
\end{align*}

In the next step, we decompose this quantum state, so that each component corresponds to one of the cases in the quantum simulator \Cref{def:quantum_simulator}. 

\paragraph*{The first query. } We start by considering the state up to the first query: $O_{H_{\vec{x}, \vec{\Theta}}} U_1 \ket 0$. We insert an additional identity operator and have,
\begin{align*}
    O_{H_{\vec{x}, \vec{\Theta}}} U_1 \ket 0 & = O_{H_{\vec{x}, \vec{\Theta}}} \, {I} \, U_1 \, \ket 0 \\
     & \numeq{i}  O_{H_{\vec{x}, \vec{\Theta}}} \left(I - \sum_{x_j} \ketbra {x_j}{x_j} + \sum_{x_j}\ketbra {x_j}{x_j}\right) U_1 \ket 0 \\
     & = O_{H_{\vec{x}, \vec{\Theta}}} \left(I - \sum_{x_j} \ketbra {x_j}{x_j} \right) U_1 \ket 0 + O_{H_{\vec{x}, \vec{\Theta}}} \left(\sum_{x_j} \ketbra {x_j}{x_j} \right) U_1 \ket 0 \\
     & \numeq{ii}  O_H \left(I - \sum_{x_j} \ketbra {x_j}{x_j} \right) U_1 \ket 0 + \sum_{x_j} O_{H_{x_j, \Theta_j}} \, \ketbra {x_j}{x_j} \,U_1 \ket 0 \\
     & = \underbrace{O_H U_1 \ket 0 \vphantom{O_{H_{x_j, \Theta_j}} \, \ketbra {x_j}{x_j}\,U_1 \ket 0}}_\textrm{(1)} - \sum_{x_j} \underbrace{O_H \ketbra {x_j}{x_j} U_1 \ket 0 \vphantom{O_{H_{x_j, \Theta_j}} \, \ketbra {x_j}{x_j}\,U_1 \ket 0}}_{\textrm{(2)}} + \sum_{x_j} \underbrace{O_{H_{x_j, \Theta_j}} \, \ketbra {x_j}{x_j}\,U_1 \ket 0}_\textrm{(3)}
\end{align*}
Above, $x_j$ is enumerated over all entries in $\vec{x}$. 

Line (i) follows easily. Line (ii) is due to the fact that, if the query input is not in $\vec{x}$, $H_{\vec{x}, \vec{\Theta}}$ is functionally equivalent to $H$; similarly, if the query input is $x_j$, $H_{\vec{x}, \vec{\Theta}}$ is functionally equivalent to $H_{x_j, \Theta_j}$. 

Next, we look at the three terms (1), (2), (3): 
\begin{itemize}
    \item[(1)] $O_H U_1 \ket 0$ corresponds to the case that no measurement happens for the first query.
    \item[(2)] $O_H \ketbra {x_j} {x_j} U_1 \ket 0$ corresponds to the case that measurement is made at the first query and the query input is $x_j$; the oracle is not reprogrammed immediately. In other words, the case $(v_1, b_1) = (1, 1)$ in the simulator. 
    \item[(3)] $O_{H_{x_j, \Theta_j}} \ketbra {x_j} {x_j} U_1 \ket 0$ corresponds to the case that measurement is made at the first query and the query input is $x_j$; the oracle is  reprogrammed immediately and used for the first query. In other words, the case $(v_1, b_1) = (1, 0)$ in the simulator. 
\end{itemize}

\paragraph*{The second query.}  We do the same: insert an additional identity operator. 
To make the presentation clearer, we focus only on one term $O_H \ketbra {x_j} {x_j} U_1 \ket 0$; the other cases are simpler.

\begin{align*}
    \, & O_{H_{\vec{x}, \vec{\Theta}}} U_2 O_H \ketbra {x_j} {x_j} U_1 \ket 0 \\
    \numeq{1}  \, & O_{H_{\vec{x}, \vec{\Theta}}} \left(I - \sum_{x_k \ne x_j} \ketbra {x_{{k}}}{x_{{k}}} + \sum_{x_k \ne x_j}\ketbra {x_{{k}}}{x_{{k}}}\right) U_2 O_H \ketbra {x_j} {x_j} U_1 \ket 0 \\
    \numeq{2} \, & O_{H_{\vec{x}, \vec{\Theta}}} \left(I - \sum_{x_k \ne x_j} \ketbra {x_{{k}}}{x_{{k}}}\right) U_2 O_H \ketbra {x_j} {x_j} U_1 \ket 0 \\
    & \quad\quad\quad\quad + \, O_{H_{\vec{x}, \vec{\Theta}}} \left(\sum_{x_k \ne x_j}\ketbra {x_{{k}}}{x_{{k}}}\right) U_2 O_H \ketbra {x_j} {x_j} U_1 \ket 0 \\
    \numeq{3} \, & O_{H_{x_j, \Theta_j}} \left(I - \sum_{x_k \ne x_j} \ketbra {x_{{k}}}{x_{{k}}}\right) U_2 O_H \ketbra {x_j} {x_j} U_1 \ket 0 \\
    & \quad\quad\quad\quad + \,  \left(\sum_{x_k \ne x_j} O_{H_{(x_j, x_k), (\Theta_j, \Theta_k)}}\ketbra {x_{{k}}}{x_{{k}}}\right) U_2 O_H \ketbra {x_j} {x_j} U_1 \ket 0 \\
    = \, & \underbrace{O_{H_{x_j, \Theta_j}}  U_2 O_H \ketbra {x_j} {x_j} U_1 \ket 0}_{(1)} {-} \sum_{x_k \ne x_j} \underbrace{O_{H_{x_j, \Theta_j}} \ketbra {x_k} {x_k}  U_2 O_H \ketbra {x_j} {x_j} U_1 \ket 0}_{(2)} \\
    &  \quad\quad\quad\quad + \, \sum_{x_k \ne x_j} \underbrace{O_{H_{(x_j, x_k), (\Theta_j, \Theta_k)}} \ketbra {x_k} {x_k}  U_2 O_H \ketbra {x_j} {x_j} U_1 \ket 0}_{(3)}
\end{align*}

We explain the equations line by line:
\begin{enumerate}
    \item This one is straightforward by realizing the summation inside the bracket is an identity operator. 
    \item This one follows from the distributive property.
    \item This is the most important one. 
    \begin{itemize}
        \item For the first term, we realize that the oracle will only be applied to inputs that are not in $\vec{x}$, or are equal to $x_j$. Thus, $H_{\vec{x}, \vec{\Theta}}$ is functionally equivalent to ${H_{x_j, \Theta_j}}$.
        \item For the second term, the oracle will only be applied to inputs that are equal to $x_k$. Thus, $H_{\vec{x}, \vec{\Theta}}$ is functionally equivalent to ${H_{(x_j, x_k), (\Theta_j, \Theta_k)}}$\footnote{It is also equivalent to ${H_{x_k,  \Theta_k}}$. However, due to our description of the simulator, $H_{(x_j, x_k), (\Theta_j, \Theta_k)}$ is more natural to work with.}.
    \end{itemize}
\end{enumerate}

\begin{itemize}
    \item[(1)] corresponds to the case that no measurement happens for the second query, but since the first query is measure-and-reprogrammed, the second query is made with the oracle $H_{x_j, \Theta_j}$. In other words, the case $(v_1, b_1) = (1,1)$. 
    \item[(2)] corresponds to the case that measurement is made at the second query and the query input is $x_k$; the oracle is not reprogrammed immediately. In other words, the case $(v_1, b_1) = (1, 1)$ and $(v_2, b_2) = (2, 1)$ in the simulator. 
    \item[(3)]  corresponds to the case that measurement is made at the second query and the query input is $x_k$; the oracle is reprogrammed immediately. In other words, the case $(v_1, b_1) = (1, 1)$ and $(v_2, b_2) = (2, 0)$ in the simulator.
\end{itemize}

\paragraph*{Generalization to all queries --- state decomposition.}
By repeating the same state decomposition up to the first $q$ queries (instead of all $q + k$ queries), we will end up a collection of subnormalized states, who sum up to the original state $\left| {\phi^{H_{\vec{x}, \vec{\Theta}}}_q} \right\rangle$. These states are parameterized by when the measurement happens (an ordered vector $\vec{v}$ such that $1 \leq v_1 \cdots \leq v_t \leq q$), whether these queries are made before or after each reprogramming ($\vec{b} \in \{0,1\}^t$), and $t \in \{0, \ldots, q\}$; in the following we will denote these states by $\ket {\phi_{\vec{v}, \vec{b}}}$. For example, assuming $\vec{b} = \vec{0}$ (all reprogramming happens immediately), we have, 
\begin{align*}
    \ket {\phi_{\vec{v}, \vec{0}}} = \sum_{\sigma \in S^k_t} &  O_{H_{\vec{x}_\sigma, \vec{\Theta}_\sigma}} U_{q} \cdots O_{H_{\vec{x}_\sigma, \vec{\Theta}_\sigma}} U_{v_t+1} \underbrace{O_{H_{\vec{x}_\sigma, \vec{\Theta}_\sigma}}\ketbra{x_{\sigma_t}}{x_{\sigma_t}} \cdots U_{v_{t-1}+1}}_{\text{stage (t)}} \\
    & \cdots \\
    & \cdot \underbrace{O_{H_{(x_{\sigma_1}, x_{\sigma_2}), (\Theta_{\sigma_1}, \Theta_{\sigma_2})}}\ketbra{x_{\sigma_2}}{x_{\sigma_2}} U_{v_2} \cdots O_{H_{x_{\sigma_1}, \Theta_{\sigma_1}}} U_{v_1+1}}_{\text{stage (2)}}  \\
    & \cdot \underbrace{O_{H_{x_{\sigma_1}, \Theta_{\sigma_1}}} \ketbra {x_{\sigma_1}}{x_{\sigma_1}}  U_{v_1} O_H \cdots O_H U_1 \ket 0}_{\text{stage (1)}} 
\end{align*}
Here $S^k_t$ denotes all ordered list of length $t$, with elements in $\{1, \ldots, k\}$ without duplication; $\vec{x}_\sigma$ denotes $(x_{\sigma_1}, \ldots, x_{\sigma_t})$ and $\vec{\Theta}_\sigma$ denotes $(\Theta_{\sigma_1}, \ldots, \Theta_{\sigma_t})$.
We can similarly define $\ket {\phi_{\vec{v}, \vec{b}}}$ for all other $\vec{b} \in \{0,1\}^t$, the only difference here is the oracle may not be immediately reprogrammed at the end of each stage. More generally, for each $\vec{v}$ of length $t$ and $\vec{b} \in \{0,1\}^t$, we define
\begin{align*}
\ket {\phi_{\vec{v}, \vec{b}}} = \sum_{\sigma \in S^k_t} \ket {\phi_{\vec{v}, \vec{b}, \sigma}},
\end{align*}
where $\ket {\phi_{\vec{v}, \vec{b}, \sigma}}$ is the state that is measured-and-reprogrammed  according to $\vec{v}, \vec{b}$ with the order $\sigma$, similar to that in the definition of $\ket {\phi_{\vec{v}, \vec{0}}}$. 
Thus, we have: 
\begin{align*}
    \left| {\phi^{H_{\vec{x}, \vec{\Theta}}}_{q}} \right\rangle = \sum_{\substack{\vec{v}, \vec{b}}} \ket {\phi_{\vec{v}, \vec{b}}},
\end{align*}

\paragraph*{Adding the extra $k$ queries.}
We assume the algorithm $\As$, after the first $q$ queries, already prepares the output $\vec{x}, z$. We will force $\As$ making the last $k$ queries, to generate $\vec{y} = H(\vec{x})$. 
Recall the definitions $\Pi^{H_{\xout, \yout}}_V$ and $G_{\xout}$ in \Cref{def:succ_prob}. By setting $\vec{x} = \xout$ and $\vec{\Theta} = \yout$ in the above analysis, the probability on the RHS in the theorem we are proving is equal to:
\begin{align*}
    \Pr_{H, G}&\left[\cA^{H_{\xout, \yout}} \rightarrow (\vec{x}, z) \text{ and } \vec{x} \equiv \xout  \text{ and } V^{H_{\xout, \yout}}(\vec{x}, H_{\xout, \yout}(\vec{x}), z) = 1\right] 
   \notag\\
    &  = \left\|G_{\xout} \Pi_V^{H_{\xout, \yout}} \left| {\phi^{H_{\xout, \yout}}_{q+k}} \right\rangle \right\|^2.
\end{align*}

Since $G_{\xout}$ and $\Pi_V^{H_{\xout, \yout}}$ commute (as they are both projections over computational basis), we can assume $G_{\xout}$ is applied to the state first. Even further, as the computation of $\vec{y} = H(\vec{x})$ and the projection $G_{\xout}$ also commute, we can assume $G_{\xout}$ applies to the state right before the last $k$ queries, which are used to compute $\vec{y}$. Therefore, for every $\ket{\phi_{\vec{v}, \vec{b}, \sigma}}$, even if $t < k$ (the length of $\vec{v}$), we can measure-and-(immediately)-reprogram exactly $k - t$ locations of the last $k$ queries, and making the random oracle exactly reprogrammed to $H_{\xout, \yout}$.

Thus, we have: 
\begin{align}\label{eq:state_decompose}
    \left| {\phi^{H_{\xout, \yout}}_{q+k}} \right\rangle = \sum_{\substack{\vec{v}, \vec{b} \\ |\vec{v}| = k}} \ket {\phi_{\vec{v}, \vec{b}}},
\end{align}
where the RHS has (at most) $2^k \binom{q + k}{k}$ terms. 

By \Cref{eq:state_decompose}, Cauchy-Schwartz and the triangle inequality, we have: 
\begin{align}
\label{eq:probability_decompose_step1}
      \left\|G_{\xout} \Pi_V^{H_{\xout, \yout}} \left| {\phi^{H_{\xout, \yout}}_{q+k}} \right\rangle \right\|^2 \leq 2^k \binom{q + k}{k} \sum_{\substack{\vec{v}, \vec{b} \\ |\vec{v}| = t} }\left\| G_{\xout} \Pi_V^{H_{\xout, \yout}} \ket {\phi_{\vec{v}, \vec{b}}} \right\|^2.
\end{align}

Finally, to prove the theorem statement, we relate each individual term on the RHS with the behaviors of our simulator $B$. 

\paragraph*{Relating each term with our simulator $B$.}
Next, we prove that each term \\
$\left\| G_{\xout} \Pi_V^{H_{\xout, \yout}} \ket {\phi_{\vec{v}, \vec{b}}} \right\|^2$ is upper bounded by the probability that when the simulator $B$ picks $\vec{v}, \vec{b}$, it succeeds and  outputs $\vec{x} \equiv \xout$, which we denote by $p_{\xout, \vec{v}, \vec{b}}$. 

Since the simulator $B$ ensures that (1) no duplicated elements ever in the control register, (2) at the end, the control register only consists of inputs that are outputted by $A$ (which will be $\xout$, enforced by $G_{\xout}$), we have that $p_{\xout, \vec{v}, \vec{b}}$ is the squared norm of the state $\left(G_{\xout} \Pi_V^{H_{\xout, \yout}} \otimes I_{\mathcal{R}} \right) \ket {\psi_{\vec{v}, \vec{b}}}$, with the state $\ket {\psi_{\vec{v}, \vec{b}}}$ being: 
\begin{align*}
    \ket {\psi_{\vec{v}, \vec{b}}} = \sum_{\sigma \in S^k_k} \ket {\phi_{\vec{v}, \vec{b}, \sigma}} \otimes \left| {\sf set} \left\{(\xo{\sigma_1}, \yo{\sigma_1}), \ldots, (\xo{\sigma_k}, \yo{\sigma_k})\right\} \right\rangle_{\mathcal{R}}.
\end{align*}
The only difference between $\ket {\psi_{\vec{v}, \vec{b}}}$ and $\ket {\phi_{\vec{v}, \vec{b}}}$ is the extra control register! However, we realize that in this case, when $\sigma$ is a permutation of $[k]$,  the control register is unentangled, making $p_{\xout, \vec{v},\vec{b}}$ is equal to $\left\| G_{\xout} \Pi_V^{H_{\xout, \yout}} \ket {\phi_{\vec{v}, \vec{b}}} \right\|^2$. This is because the set will simply be ${\sf set}\{(\xo{1}, \yo{1}), \ldots, (\xo{k}, \yo{k})\}$, regardless of what $\sigma$ is.

Finally, we have the L.H.S. is equal to
\begin{align*}
& \Pr_{H, G}\left[\Sim^{H, G, A} \rightarrow (\vec{x}, \vec{y}, z)  \text{ and } \vec{x} \equiv \xout \text{ and } V^{H_{\xout, \yout}}(\vec{x}, \vec{y}, z) = 1\right] \\
= & \frac{1}{2^k \binom{q+k}{k}} \sum_{\substack{\vec{v}, \vec{b} \\ |\vec{v}| = k}} p_{\xout, \vec{v}, \vec{b}}
\end{align*}

Thus, combining \Cref{eq:probability_decompose_step1} and the equation above, we have: 
\begin{align*}
    \text{R.H.S.} = \left\|G_{\xout} \Pi_V^{H_{\xout, \yout}} \left| {\phi^{H_{\xout, \yout}}_q}\right\rangle \right\|^2   & \leq 2^k \binom{q+k}{k} \sum_{\substack{\vec{v}, \vec{b} \\ |\vec{v}| = k}} p_{\xout, \vec{v}, \vec{b}} \\
       & \leq \left(2^k \binom{q+k}{k}\right)^2 \cdot \frac{1}{2^k \binom{q+k}{k}} \sum_{\substack{\vec{v}, \vec{b} \\ |\vec{v}| = k}} p_{\xout, \vec{v}, \vec{b}} \\
    & = \text{ L.H.S.}
\end{align*}
where L.H.S. and R.H.S. denote the left/right-hand side term in the theorem statement.
Therefore, we conclude the proof.

\end{proof}

\begin{lemma}[Coherent Measure-and-Reprogram results in Uniform Images]
\label{lemma:uniform_images}
Consider the Coherent Measure-and-Reprogram Experiment in Definition~\ref{def:quantum_simulator}, {but where we choose $G$ to be uniformly random} Then, for the measure-and-reprogram algorithm $\cB$, the measurement $L = (\vec{x'}, \vec{\Theta'})$ of the $\cR$ register (in Step 5) will result in uniformly random images $\vec{\Theta'}$.
\end{lemma}
\begin{proof}
We will proceed with a proof by induction over the number of quantum queries of $\cA$. In this proof, we will denote by $n$ the total number of queries ($n = q + k$).
For $n = 1$, let $v_j = n = 1$. Then, if $b_j = 0$,
after updating the register $\cR$ using the unitary $U$ (in step 2), the register $\cR$ will contain superposition of $L$ sets consisting of a single pair $(x, G(x))$.
Then, we perform the query to $H$ controlled by ${\cal R}$ using $O_H^{\sf ctrl}$, which is a query to the reprogrammed $H$ on single points $x$, modifying accordingly the image register ($y \rightarrow y \oplus H_{x, G(x)}$), but
which does not affect the $\cR$ register. As $G$ is random oracle, measuring $\cR$ will result in a uniform image $\theta' = G(x)$, for some $x \in X$. If $b_j = 1$, we first query using $O_H^{\sf ctrl}$, which is a query to the original $H$ as $L$ is empty. Then, we update $\cR$ using unitary $U$, resulting in $\cR$ containing superposition of $L$ sets consisting of a single pair $(x, G(x))$. As before, measuring $\cR$ will result in a uniform image $\theta' = G(x)$, for some $x \in X$. {We emphasize that although Definition~\ref{def:quantum_simulator} defines $G$ as an arbitrary function, in the statement
of this Lemma, we consider uniformly random $G$ instead of an arbitrary $G$.}

For the inductive step, suppose that up to query $n - 1$, the register $\cR$ consists of sets $L'$ with uniform images.
Let $\cA$ make its $n$-th query. If there does not exist any $v_j$ equal to $n$ then algorithm $\cB$ answer $\cA$'s query controlled by $\cR$, reprogramming the oracle with the inputs and outputs pairs in $L'$, but importantly $\cR$ remains unchanged, hence $\cR$ contains only uniform images by our inductive hypothesis. Otherwise, suppose there exists $j^*$ such that $v_{j^*} = n$. 
Then if $b_{j^*} = 0$, we are first going to add in $L$ the pair $(x, G(x))$ if $x$ is not already in $L$, i.e. $L = L' \cup \{(x, G(x))\}$, otherwise $L = L'$. We are then going to make the controlled query $O_H^{\sf ctrl}$ to the reprogrammed oracle $O_{H_L}$, which does not affect the register $\cR$. Hence by measuring $\cR$ results in either $(x', \theta') \in L'$, which by hypothesis contains uniform image $\theta'$ or in $(x, G(x))$, which given that $G$ is a random oracle, also results in a uniform image.
Similarly, if $b_{j^*} = 1$ we are first going to make the controlled query $O_H^{\sf ctrl}$ to the reprogrammed oracle $O_{H_L'}$, then we will update the register $\cR$ using the unitary $U$, which as before will either result in either $L = L' \cup \{(x, G(x))\}$ or $L = L'$. In both cases, by measuring $\cR$ we will get a uniform image by using the uniformity of $G$ and the inductive hypothesis.
\end{proof}

\section{Applications}

\subsection{Query Complexity}

We will begin by first introducing the family of (security) games for which we will establish their quantum query complexity, namely the hardness of a quantum adversary to win such games.

\begin{definition}[Multi-Output $k$-Search Game (Single-Instance)] 
    \label{def:multi_instance_game}
Let the random oracle $H : [M] \rightarrow [N]$,  a distribution over challenges $\pi_H$ and a winning relation $R_{H, ch}$ defined over $Y^k$. \\
Then we define the \multigame $\gamemath$ as follows:
\begin{enumerate}
    \item Challenger samples randomness $ch$ and sends it to a quantum algorithm $\cA$ having (quantum) oracle access to $H$;
    \item Adversary $\cA$ outputs 
    $\vec{x} := (x_1, ..., x_k), z$;
    \item Challenger queries $\vec{x}$ to the random oracle, resulting in $\vec{y} := (y_1 = H(x_1), ..., y_k = H(x_k))$ and
    checks if they satisfy the winning relation: \\
    $b := ((x_1, \ldots, x_k, y_1, \ldots, y_k, z) \in R_{H, ch})$;
    \item If $b = 1$, $\cA$ wins the $\gamemath$ game.
\end{enumerate}
We will denote by $\epsilon_{\gamemath}(q)$ the maximum probability over all $q$-quantum algorithms $\cA$ of winning the \multigame $\gamemath$.
\end{definition}

Our main result is a quantum lifting theorem in the average case, relating the success probability of an arbitrary quantum algorithm to win a \multigame with the probability of success of a quantum algorithm equipped with exactly $k$ quantum queries.

\begin{theorem}[Lifting for Multi-Output $k$-Search Games]
\label{lem:quantum_lifting}
Let $\gamemath$ be a 
\multigame
(as defined in Def.~\ref{def:multi_instance_game}).
Let $\cA$ be a $q$-quantum query adversary in the game $\gamemath$ (against the $k$-classical query challenger $\cC$).
Then there exists a $k$-quantum query adversary $\cB$ against the game such that:
\begin{equation*}
\Pr[\cB^{\ket{H}} \text{ wins } \gamemath] \geq \frac{1}{{2^{2k}{{q + k} \choose k}^2}} \Pr[\cA^{\ket{H}} \text{ wins } \gamemath]. 
\end{equation*}
\end{theorem}
\begin{proof}
We will show that our Coherent Reprogramming result in Theorem~\ref{thm:coherent_measure_and_reprogram} implies the lifting theorem. We will now show how to instantiate the coherent reprogramming theorem. Let $\xout$ be uniformly sampled from $X^k$. Let $H', G' : X \rightarrow Y$ be two uniform random oracles. Then, it is clear that, as $\yout = G'(\xout)$ is also uniform over $Y^k$, the reprogrammed function ${H'}_{\xout, \yout} : X \rightarrow Y$ is a uniform random function; this is due to the fact that, as stated in Theorem~\ref{thm:coherent_measure_and_reprogram} (invoked here), the tuple $\xout$ has distinct values for its element.
We will instantiate the random oracle in the game $\gamemath$ as the function ${H'}_{\xout, \yout}$. Assume that in the game $\gamemath$ after receiving the challenge and after performing its $q$ quantum queries to ${H'}_{\xout, \yout}$, the adversary $\cA$ returns to the Challenger the outcome $\vec{x}$. Then, the Challenger queries $\vec{x}$ to ${H'}_{\xout, \yout}$ resulting in $\vec{y}$ and checks if $\vec{y}$ satisfies the winning relation $R_{{H'}_{\xout, \yout}, ch}$. Define $V^{{H'}_{\xout, \yout}}$ as the predicate that outputs $1$ if $\vec{y} \in R_{{H'}_{\xout, \yout}, ch}$ and $0$ else. In this way, we observe that the probability that $\cA$ wins the game $\gamemath$ is exactly the RHS of Theorem~\ref{thm:coherent_measure_and_reprogram}. As a result, by Theorem~\ref{thm:coherent_measure_and_reprogram}, there must exist an efficient quantum simulator $\Sim^{H', G', \cA}$ performing $k$ quantum queries that also wins the game $\gamemath$. Hence, it suffices to instantiate $\cB$ as the simulator $\Sim$.

\end{proof}

Let $L_{\cC}$ represent the set of (classical) queries that a challenger performs during a \multigame $\gamemath$ (Def.~\ref{def:multi_instance_game}). For a quantum query adversary $\cB$ against $\gamemath$, we will denote by $L_{\cB}$ the result of measuring its input and output query registers.
Now, for the query complexity applications we will need the following stronger lifting theorem, which intuitively additionally guarantees the existence of an algorithm against $\gamemath$ such that at the end of the game, measuring its input and output registers gives us exactly the set of queries of the challenger.

\begin{theorem}[Lifting for Search Game with Uniform Images]
\label{lem:quantum_lifting_improved}
Let $\gamemath$ be a 
\multigame (as defined in Def.~\ref{def:multi_instance_game}).
Let $\cA$ be a $q$-quantum query adversary in the game $\gamemath$ (against the $k$-classical query challenger $\cC$).
Then there exists a $k$-quantum query adversary $\cB$ such that $L_{\cB}$ is uniform, satisfying:
\begin{equation*}
    \Pr[\cB^{\ket{H}} \text{ wins } \gamemath \text{ and } L_{\cC} = L_{\cB}] \geq \frac{1}{{2^{2k}{{q + k} \choose k}^2}} \Pr[\cA^{\ket{H}} \text{ wins } \gamemath].
\end{equation*}
\end{theorem}
\begin{proof}
     The simulator algorithm $\cB$ will follow the outline of the algorithm in the proof of Theorem~\ref{lem:quantum_lifting}, with the only difference that $\cB$ will perform an additional step at the end. 
Namely, after interaction with Challenger $\cC$, compute list of queries of $\cC$ as $L_{\cC}$. If any query in $L_{\cC}$ has not yet been queried by $\cB$, $\cB$ will query them to oracle $H$. The uniformity of $L_{\cB}$ follows directly from Lemma~\ref{lemma:uniform_images}.
\end{proof}

\subsection{A New Quantum Lifting Theorem and Direct Product Theorem for Image Relations}

Our first quantum lifting result (in Theorem~\ref{lem:quantum_lifting}) gives a first bound on the quantum hardness of solving any \multigame $\gamemath$ by relating it to the probability of $\gamemath$ being solved by a quantum algorithm with a small number of quantum queries.
In this section we can derive a stronger quantum lifting theorem for the class of relations that only depend on images.
\begin{theorem}
[Quantum Lifting Theorem for Image Relations]
\label{thm:direct_product_strong}
    For any quantum algorithm $\cA$ equipped with $q$ quantum queries, $\cA$'s success probability to solve the \multigame specified by the winning relation $R$, is bounded by:
    \begin{equation*}
        \begin{split}
        \Pr[\cA^{\ket{H}} & \text{ wins \multigame}] \leq \\
        &{2^{2k}{{q + k} \choose k}^2} \Pr[\exists \text{ perm } \pi \ | \ (y_{\pi(1)}, y_{\pi(2)}, ..., y_{\pi(k)}) \in R  : (y_1, ..., y_k) \xleftarrow{\$} Y^k].
        \end{split}
    \end{equation*}

For simplicity, in the rest of the section, we define $p(R)$ as:
\begin{align*}
    p(R) = \Pr[\exists \text{ perm } \pi \ | \ (y_{\pi(1)}, y_{\pi(2)}, ..., y_{\pi(k)}) \in R  : (y_1, ..., y_k) \xleftarrow{\$} Y^k].
\end{align*}
\end{theorem}
\begin{proof}

    Let $\gamemath$ be a \multigame and assume a $q$-quantum adversary $\cA$ sends to the Challenger the answer $\vec{x} = (x_1, ..., x_k)$. Challenger $\cC$ will accept if and only if $\vec{y} := (H(x_1), ..., H(x_k)) \in R_{H, ch}$ and if $x_i, x_j$ are pairwise distinct.
    By Theorem~\ref{lem:quantum_lifting_improved} we know there exists a quantum algorithm $\cB$ making $k$ quantum queries to $H$ winning the game such that $L_{\cB} = L_{\cC}$ with success probability at least the success probability of $\cA$ multiplied by a factor of
    ${2^{2k}{{q + k} \choose k}^2}$. The condition $L_{\cB} = L_{\cC}$ implies that $\cC$ will verify as the images of $\cB$'s answer exactly a permutation of the recorded information in $L_{\cB}$. Therefore, due to the property of Theorem~\ref{lem:quantum_lifting_improved} that $L_{\cB}$ will be uniformly over $Y^k$, $\cB$'s winning probability will be lower bounded by the probability that there exists a permutation such that for uniformly sampled images from $Y^k$, the permuted images will belong to our target relation:
    \begin{align*}
    \Pr[\cA^{\ket{H}} & \text{ wins \multigame}] \leq 
                {2^{2k}{{q + k} \choose k}^2} p(R) \, .
    \end{align*}
\end{proof}

Next, we show a Direct Product Theorem for Image Relations. 
\begin{definition}[Direct Product]
    Let $\gamemath$ be a multi-output $k$-search game specified by the winning relation $R$, with respect to a random oracle $[M] \to [N]$. Define the following Direct Product $\gamemath^{\otimes g}$: 
    \begin{itemize}
        \item Let $H$ be a random oracle $[g] \times [M] \to [N]$, and $H_i$ denotes $H(i, \cdot)$;  
        \item Challenger samples $ch_i$ as in $\gamemath$ for $i \in \{1, \ldots, g\}$. 
        \item Adversary $\As$ gets oracle access to $H$ and outputs $\vec{x}_1, \ldots, \vec{x}_g$, $z_1, \ldots, z_g$ such that each input in $\vec{x}_i$ start with $i$.  
        \item Challenger computes $b_i := (\vec{x}_i, H(\vec{x}_i), z_i) \in R_{H_i, ch_i}$;
        \item If all $b_i$ equal to $1$, $\As$ wins the $\gamemath^{\otimes g}$ game. 
    \end{itemize}
\end{definition}

\begin{theorem}[Direct Product Theorem for Image Relations]\label{thm:dpt}
    For any quantum algorithm $\As$ equipped with $g q$ quantum queries, $\As$'s success probability to solve the Direct Product $\gamemath^{\otimes g}$ with the underlying $\gamemath$ specified by the winning relation $R$, is bounded by
\begin{align*}
    \Pr[\As^{|H\rangle} & \text{ wins } G^{\otimes g}] \leq \left( {2^{2k}{{q + k} \choose k}^2} p(R) \right)^g.
\end{align*}
\end{theorem}
\begin{proof}
Let $\gamemath$ be a \multigame and assume a $g q$-quantum adversary $\cA$ sends to the Challenger the answer $\vec{x}_1, \ldots, \vec{x}_g, z_1, \ldots, z_g$. 
By Theorem~\ref{lem:quantum_lifting_improved} we know there exists a quantum algorithm $\cB$ making $gk$ quantum queries to $H$ winning the game such that $L_{\cB} = L_{\cC}$ with success probability at least the success probability of $\cA$ multiplied by a factor of
    ${2^{2gk}{{gq + gk} \choose gk}^2}$. The condition $L_{\cB} = L_{\cC}$ implies that $\cC$ will verify as the images of $\cB$'s answer exactly a permutation of the recorded information in $L_{\cB}$. Moreover, for every image $y$, its associated input $x$ only belongs to one of the oracles $H(i, \cdot)$; thus, it can only contribute to one of the relation checks $R_{H_i, ch_i}$. Thus, the permutation of the recorded information can only permute images with respect to the same oracle $H_i$.

    Therefore, due to the property of Theorem~\ref{lem:quantum_lifting_improved} that $L_{\cB}$ will be uniformly over $Y^{gk}$, $\cB$'s winning probability will be lower bounded by the probability that there exists a permutation such that for uniformly sampled images from $Y^{gk}$, the permuted images will belong to our target relation:
    \begin{align*}               & \Pr[\cA^{\ket{H}} \text{ wins } \gamemath^{\otimes g}] \\ & \leq {2^{2gk}{{gq + gk} \choose gk}^2}\Pr[\exists  \pi_1,\ldots,\pi_g \in \sym_k \ | \ (y_{i, {\pi_i(1)}}, ..., y_{i, {\pi_i(k)}}) \in R  : (y_{i, 1}, ..., y_{i, k}) \xleftarrow{\$} Y^k] \\
                & \leq  \left( {2^{2k}{{q + k} \choose k}^2} \Pr[\exists \text{  } \pi \in \sym_k \ | \ (y_{\pi(1)}, y_{\pi(2)}, ..., y_{\pi(k)}) \in R  : (y_1, ..., y_k) \xleftarrow{\$} Y^k] \right)^g \\
                & \leq  \left( {2^{2k}{{q + k} \choose k}^2} p(R)\right)^g \, .
    \end{align*}    
\end{proof}

In the following section, we will show some of the query complexity and cryptographic applications of our quantum lifting theorems and Direct Product Theorem. 

\subsubsection{Application 1: Non-uniform Security}

\begin{definition}[Advice Algorithms]
We define an advice (non-uniform) algorithm $\cA = (\cA_1, \cA_2)$ equipped with $q$ queries and advice of length $S$ as follows:
\begin{enumerate}
    \item $\cA_1^{H} \rightarrow \ket{adv}$: an unbounded algorithm $\cA_1$ outputs the advice $\ket{adv}$ consisting of $S$ qubits;
    \item $\cA_2^{H}(\ket{adv}, ch) \rightarrow x$: $q$-quantum algorithm $\cA_2$ takes as input the quantum advice $\ket{adv}$ and a challenge $ch$, outputs answer $x$;
\end{enumerate}
    We define $\epsilon_{\gamemath}^C(q, S)$ as the maximum winning probability over any advice adversary $\cA$ equipped with $q$ quantum queries and $S$ classical bits of advice against the classically-verifiable search game $\gamemath$.
\end{definition}

We also consider multi-instance games, similar to Direct Product, except all the instances share the same oracle. 
\begin{definition}[Multi-Instance Game]
    Let $\gamemath$ be a multi-output $k$-search game specified by the winning relation $R$, with respect to a random oracle $[M] \to [N]$. Define the following Direct Product $\gamemath_{\sf MIS}^{\otimes g}$: 
    \begin{itemize}
        \item Let $H$ be a random oracle $[M] \to [N]$; 
        \item Challenger samples $ch_i$ as in $\gamemath$ for $i \in \{1, \ldots, g\}$;
        \item Adversary $\As$ gets oracle access to $H$ and outputs $\vec{x}_1, \ldots, \vec{x}_g$, $z_1, \ldots, z_g$; 
        \item Challenger computes $b_i := (\vec{x}_i, H(\vec{x}_i), z_i) \in R_{H, ch_i}$;
        \item If all $b_i$ equal to $1$, $\As$ wins the $\gamemath_{\sf MIS}^{\otimes g}$ game. 
    \end{itemize}
    From above, we can define $R^{\otimes g}_{\sf MIS}$ as the winning relation for $\gamemath_{\sf MIS}^{\otimes g}$.
\end{definition}

\begin{lemma}[Multi-Output Implies Non-Uniform Classical Advice (\cite{chung2020tight})]
\label{lemma:salting_classical}
    Let $\gamemath$ be a search game (as defined in Def.~\ref{def:multi_instance_game}).
    If the maximum winning probability for any quantum algorithm equipped with $q$ quantum queries against $\gamemath^{\otimes g}_{\sf MIS}$ is $\epsilon_{\gamemath_{\sf MIS}^{\otimes g}}(q)$, then the maximum probability of any non-uniform adversary equipped with $q$ quantum queries and $S$-length classical advice against the original game $\gamemath$ is at most:
    \begin{equation*}
        \epsilon_{\gamemath}^C(q, S) \leq 4 \cdot \left[\epsilon_{\gamemath_{\sf MIS}^{\otimes {S}}}(Sq)\right]^{\frac{1}{S}}
    \end{equation*}
\end{lemma}

By combining these results with our Quantum Lifting Theorem, we derive the security against advice (non-uniform) quantum algorithms.

\begin{lemma}[Security against Advice Quantum Adversaries]
\label{lemma:advice_quantum}
Let $\gamemath$ be any multi-output $k$-search game specified by the winning relation $R$. Let $\gamemath^{\otimes g}_{\sf MIS}$ be the multi-instance game and $R^{\otimes g}_{\sf MIS}$ be the relation. 
Any non-uniform algorithm $\cA$ equipped with $q$ quantum queries and $S$ classical bits of advice can win the game $\gamemath$ with probability at most:

\begin{align*}   \epsilon_{\gamemath}^C(q, S) & \leq 4 \cdot 2^{2k} {{S(q+k)} \choose {Sk}}^{\frac{2}{S}} \cdot p(R_{\sf MIS}^{\otimes S}) \, .
\end{align*}

\end{lemma}
\begin{proof}
By combining our (strong) quantum lifting theorem (in Theorem~\ref{thm:direct_product_strong}) with the two advice results (Lemma~\ref{lemma:salting_classical}).
\end{proof}

\subsubsection{Application 2: Salting Against Non-uniform Adversaries}

\begin{definition} [Salted Game] Let $\gamemath$ be a search game (as defined in Def.~\ref{def:multi_instance_game}) specified by a random oracle $H : [M] \rightarrow [N]$, a distribution over challenges $\pi_H$ and a winning relation $R_{H, ch}$ defined over $Y$. Then we define the salted version of $\gamemath$ as the game $\gamemath_s$ with salted space $[K]$ defined as follows:
\begin{enumerate}
    \item The random oracle function is defined as: $G = (H_1, ..., H_K)$ for $K$ random functions $H_i : [M] \rightarrow [N]$;
    \item For any such $G$, the challenge $ch := (i, ch_i)$ is produced by first sampling uniformly at random $i \in [K]$ and then sampling $ch_i$ according to $\pi_{H_i}$;
    \item The winning relation is defined as $R_{G, ch} := R_{H_i, ch_i}$;
\end{enumerate}
We will denote by $\epsilon_{\gamemath_s}(q)$ the maximum probability over all $q$-quantum algorithms $\cA$ of winning the salted game $\gamemath_s$. 
\end{definition}

\begin{lemma}[Security of Salted Game against Classical Advice]    
\label{lemma:salted_multi_advice}
Let $\gamemath$ be a multi-output $k$-search game (as defined in Def.~\ref{def:multi_instance_game}), specified by a relation $R$. Let $\gamemath_s$ be the salted game, with salt space $[K]$. Then we have,
\begin{align*}
    \epsilon^C_{\gamemath_s}(q, S) \leq 4 \cdot \frac{S}{K} + 4 \cdot {2^{2k}{{q + k} \choose k}^2} p(R). 
\end{align*}
\end{lemma}
\begin{proof}
By \Cref{lemma:salting_classical}, the non-uniform security is related to the multi-instance game $\gamemath_{s, \sf MIS}^{\otimes g}$, with salt space $[K]$. The security of the multi-instance game is closely related to the Direct Product, for salted games, as shown in \cite{dong2024salting} (in the proof of Theorem 4.1). More precisely, 
\begin{align*}
    \epsilon_{\gamemath_{s, \sf MIS}^{\otimes g}}(g q)^{1/g} \leq \epsilon_{\gamemath_{s}^{\otimes g}}(g q)^{1/g} + \frac{g}{K}. 
\end{align*}
Intuitively, the only difference between the multi-instance game and the Direct Product is that, the same salt can be sampled with duplication. The extra factor $\frac{g}{K}$ captures the fact that the salt can be duplicated. 
Combining with \Cref{thm:dpt}, we have:
\begin{align*}
    \epsilon_{\gamemath_{s}}^C(q, S) & \leq 4 \left( \epsilon_{\gamemath_{s, \sf MIS}^{\otimes S}} (S q) \right)^{1/S} \\
     & \leq 4 \left(\epsilon_{\gamemath_{s}^{\otimes S}}(S q)^{1/S} + \frac{S}{K} \right) \\
     & \leq 4 \cdot \frac{S}{K} + 4 \cdot 2^{2k} \binom{q+k}{k}^2 p(R). 
\end{align*}

\end{proof}

\subsubsection{Application 3: Multi-Image Inversion}

\hfill \break

Our first result establishes the quantum hardness of multi-image inversion, which is a tight bound as already proven in \cite{chung2020tight}, but achieved here in a much simpler way, directly from our quantum lifting theorem.

\begin{lemma}[Quantum Hardness of Multi-Image Inversion]
\label{lemma:hardness_multi_image}

    For any $\vec{y} = (y_1, ..., y_k) \in Y^k = [N]^k$ (without duplicates) and for any $q$-quantum query algorithm $\cA$ whose task is to invert all the images in $\vec{y}$, the success probability of $\cA$ is upper bounded by:
    \begin{align*}
         & \Pr_H[\cA^{\ket{H}}(\vec{y}) \rightarrow \vec{x} = (x_1, ..., x_k) \ : \ H(x_i) = y_i \ \forall i \in [k]] \nonumber \\
             \leq &               \left[ \frac{4e(q + k)^2}{Nk} \right]^k
    \end{align*}
\end{lemma}

\begin{proof}
We will show this using our strong quantum lifting theorem for image relations (Theorem~\ref{thm:direct_product_strong}).
Define $R$ as the relation over $[N]^k$, with $H: [M] \rightarrow [N]$ such that: $R = \{y_1, ..., y_k\}$. 
Then for each permutation $\pi$, we have $\Pr[(y_{\pi(1)}, ..., y_{\pi(k)}) \in R \ | \ (y_1, ..., y_k) \leftarrow [N]^k] = \frac{1}{N^k}$. Using that the number of permutations $\pi$ is $k!$ leads to:
\begin{align*}
    \Pr_H[\cA^{\ket{H}}(\vec{y}) \rightarrow \vec{x} = (x_1, ..., x_k) \ : \ H(x_i) = y_i \ \forall i \in [k]] \leq {2^{2k}{{q + k} \choose k}^2} \cdot \frac{k!}{N^k}
\end{align*}
Using first the inequality ${{q + k} \choose k} \leq \frac{(q + k)^k}{k!}$ and then the Stirling approximation $k! \geq \left(\frac{k}{e}\right)^k$, we get:
\begin{align*}
    {2^{2k}{{q + k} \choose k}^2} \cdot \frac{k!}{N^k} &\leq \left(\frac{4}{N}\right)^k \cdot k! \cdot \left[\frac{(q + k)^k}{k!}\right]^2 \\
    &\leq \left(\frac{4}{N}\right)^k \cdot (q + k)^{2k} \cdot \left(\frac{e}{k}\right)^k = \left[ \frac{4e(q + k)^2}{Nk} \right]^k 
\end{align*}
\end{proof}

\subsubsection{Application 4: Multi-Collision Finding and Multi-Search}

\hfill \break

Next, we can determine the quantum hardness of the multi-collision problem, namely finding $k$ different inputs that map to the same output of the random oracle.

\begin{lemma}[Quantum Hardness of Multi-Collision Finding and Salted Multi-Collision Finding]
\label{lemma:hardness_multi_collision}

    For any $q$-quantum query algorithm $\cA$, we have the upper bound for solving the $k$-multi-collision problem:
        \begin{equation*}
        \Pr_H[\cA^{\ket{H}}() \rightarrow \vec{x} = (x_1, ..., x_k) \ : \ H(x_1) = ... = H(x_k)] \leq 
        \frac{1}{N^{k - 1}} \left[\frac{2e(q+k)}{k}\right]^{2k}
    \end{equation*}
     Any quantum algorithm $\cA$ equipped with $q$ quantum queries and $S$-bit of classical advice can win the salted multi-collision finding game with salted space $[K]$ with probability at most:
 \begin{equation*}
     \Pr_H[\cA() \rightarrow \vec{x} = (x_1, ..., x_k) \ : \ H(x_1) = ... = H(x_k)] \leq \frac{4}{N^{k - 1}} \left[\frac{2e(q+k)}{k}\right]^{2k} + \frac{4S}{K}.
 \end{equation*}
\end{lemma}
\begin{proof}
We will show this using our strong quantum lifting theorem for image relations (Theorem~\ref{thm:direct_product_strong}).
Define $R := \{y, ..., y\}_y$ the relation over $[N]^k$, where $H: [M] \rightarrow [N]$. Then for each permutation $\pi$, we have $\Pr[(y_{\pi(1)}, ..., y_{\pi(k)}) \in R \ | \ (y_1, ..., y_k) \leftarrow [N]^k] = \frac{1}{N^{k-1}}$. As $R$ is permutation invariant, this implies that:
    \begin{equation*}
        \Pr_H[\cA^{\ket{H}}() \rightarrow \vec{x} = (x_1, ..., x_k) \ : \ H(x_1) = ... = H(x_k)] \leq 
         {2^{2k}{{q + k} \choose k}^2} \cdot \frac{1}{N^{k - 1}}
    \end{equation*}
    Using first the inequality ${{q + k} \choose k} \leq \frac{(q + k)^k}{k!}$ and then the Stirling approximation:
\begin{align*}
    {2^{2k}{{q + k} \choose k}^2} \cdot \frac{1}{N^{k-1}} &\leq N \cdot \left(\frac{4}{N}\right)^k \cdot \left[\frac{(q + k)^k}{k!}\right]^2 \\
    &\leq N \left(\frac{4}{N}\right)^k \cdot (q + k)^{2k} \cdot \left(\frac{e}{k}\right)^{2k} 
    = \frac{1}{N^{k - 1}} \left[\frac{2e(q+k)}{k}\right]^{2k}
\end{align*}
 Finally, the security of salted multi-collision against non-uniform quantum adversaries equipped with $S$ bits of advice follows by combining this quantum hardness bound of multi-collision with Lemma~\ref{lemma:salted_multi_advice}.
\end{proof}

Finally, we consider another search application, namely the task of determining $k$ different inputs that all map to $0$ under the random oracle. One of the main motivations behind this problem is its relation to the notion of proof-of-work in the blockchain context~\cite{garay2015bitcoin}.

\begin{lemma}[Quantum Hardness of Multi-Search]
    For any $q$-quantum query algorithm $\cA$ whose task is to find different preimages of $0$ of a random oracle $H$, the success probability of $\cA$ is upper bounded by:
       \begin{align*}
         \Pr_H[\cA^{\ket{H}}() \rightarrow \vec{x} = (x_1, ..., x_k) \ : \ H(x_i) = 0 \ \forall i \in [k]]  
             \leq              \left[\frac{4e^2 (q + k)^2}{Nk^2} \right]^k
    \end{align*}
\end{lemma}
Note that this bound is asymptotically tight, as an algorithm with $q$ queries can use $q/k$ queries to find each pre-image (Grover's algorithm), resulting in a probability of $\Theta\left( \left\{ (\frac{q}{k})^2/N \right\}^k\right)$. 
\begin{proof}
    We will show this using our strong quantum lifting theorem for image relations (Theorem~\ref{thm:direct_product_strong}).
Define $R := \{0, ..., 0\}$ the relation over $[N]^k$, where $H: [M] \rightarrow [N]$. Then for each permutation $\pi$, we have $\Pr[(y_{\pi(1)}, ..., y_{\pi(k)}) \in R \ | \ (y_1, ..., y_k) \leftarrow [N]^k] = \frac{1}{N^{k}}$. As $R$ is permutation invariant, this implies that:
  \begin{align*}
    \Pr_H[\cA^{\ket{H}}() \rightarrow \vec{x} = (x_1, ..., x_k) \ : \ H(x_i) = 0 \ \forall i \in [k]] 
             \leq 2^{2k} \cdot {{q + k} \choose k}^2 \frac{1}{N^k}
    \end{align*}
    Using first the inequality ${{q + k} \choose k} \leq \frac{(q + k)^k}{k!}$ and then the Stirling approximation:
\begin{align*}
    {2^{2k}{{q + k} \choose k}^2} \cdot \frac{1}{N^k} &\leq \left(\frac{4}{N}\right)^k \cdot \left[\frac{(q + k)^k}{k!}\right]^2 \\
    &\leq \left(\frac{4}{N}\right)^k \cdot (q + k)^{2k} \cdot \left(\frac{e}{k}\right)^{2k} 
   = \left[\frac{4e^2 (q + k)^2}{Nk^2} \right]^k
\end{align*}
\end{proof}

\section*{{Acknowledgements}}
J.G. was partially supported by NSF SaTC grants no. 2001082 and 2055694.
F.S. was partially supported by NSF grant no. 1942706 (CAREER). J.G. and F.S. were also partially support by Sony by means of the Sony Research Award Program.
A.C. acknowledges support from the National Science Foundation grant CCF-1813814, from the AFOSR under Award Number FA9550-20-1-0108 and the support of the Quantum Advantage Pathfinder (QAP) project, with grant reference EP/X026167/1 and the UK Engineering and Physical Sciences Research Council.

\printbibliography

\end{document}